\documentclass[a4paper,12pt]{article}
\usepackage{mathrsfs}
\usepackage{}
\usepackage{subfigure}
\usepackage[top=2.5cm,bottom=2.5cm,left=2.5cm,right=2.5cm]{geometry}
\usepackage{amssymb}
\usepackage{graphicx}
\usepackage{amsmath,amsthm,amssymb,lineno}
\usepackage{latexsym}
\usepackage{graphicx,booktabs,multirow}
\usepackage{latexsym, tabularx,shapepar}
\usepackage[all,2cell,dvips]{xy} \UseAllTwocells \SilentMatrices
\usepackage{appendix}
\usepackage{longtable}
\newtheorem{defi}{Definition}[section]
\newtheorem{teo}[defi]{Theorem}

\newtheorem{cor}[defi]{Corollary}

\newtheorem{es}[defi]{Example}

\begin{document}
\title{Consensus dynamics and coherence in hierarchical small-world networks}
\author{Yunhua Liao$^{1,2}$, Mohamed Maama$^{3}$ and M.A. Aziz-Alaoui$^{4}$\\
$^1$ Department of Mathematics, Hunan University of Technology\\
and Business, Changsha Hunan 410205, China\\
$^2$ Key Laboratory of Hunan Province for Statistical Learning \\
and Intelligent Computation, Changsha Hunan 410205, China\\
(307156168@qq.com) \\
$^3$ Applied Mathematics and Computational Science Program, \\ KAUST, Thuwal, 23955-6900, KSA\\
(maama.mohamed@gmail.com) \\
$^4$ Normandie Univ, UNIHAVRE, LMAH, FR-CNRS-3335, \\ ISCN, Le Havre 76600, France\\
(aziz.alaoui@univ-lehavre.fr)
}

\date{}
\maketitle
\begin{abstract}
The hierarchical small-world network is a real-world network. It models well the benefit transmission web of the pyramid selling in China and many other countries. In this paper, by applying the spectral graph theory, we study three important aspects of the consensus problem in the hierarchical small-world network: convergence speed, communication time-delay robustness, and network coherence. Firstly, we explicitly determine the Laplacian eigenvalues of the hierarchical small-world network by making use of its treelike structure. Secondly, we find that the consensus algorithm on the hierarchical small-world network converges faster than that on some well-studied sparse networks, but is less robust to time delay. The closed-form of the first-order and the second-order network coherence are also derived. Our result shows that the hierarchical small-world network has an optimal structure of noisy consensus dynamics. Therefore, we provide a positive answer to two open questions of Yi \emph{et al}. Finally, we argue that some network structure characteristics, such as large maximum degree, small average path length, and large vertex and edge connectivity, are responsible for the strong robustness with respect to external perturbations.
\end{abstract}

\baselineskip=0.30in

\section{Introduction}
The consensus problem has been primarily investigated in management science and statistics~\cite{DeGroot74}. And now, it is a challenging and hot research area for multiagent systems~\cite{Saber04}. In these settings, consensus means that all agents reach an agreement on one common issue. Consensus problems have emerged in various disciplines, ranging from distributed computing~\cite{Korniss03}, sensor networks~\cite{Li06}, biological systems~\cite{Sumpter08,Nabet09} to human group dynamics~\cite{Giraldo16}. Due to their broad applications, consensus problems have attracted considerable attention in recent years~\cite{WXD19,WXM19,DZHL20}.

Convergence speed, communication time-delay robustness, and network coherence are three primary aspects of analysing the consensus protocol. Convergence speed measures the time of convergence of the consensus algorithm. It was proved that convergence speed is determined by the second smallest Laplacian eigenvalue $\lambda_2$~\cite{Saber04,Qi18}. Communication time-delay robustness refers to the ability of the consensus algorithm resistant to communication delay between agents, with the allowable maximum delay determined by the largest Laplacian eigenvalue $\lambda_n$~\cite{Saber04,Qi18}. Network coherence quantifies the robustness of the consensus algorithm to stochastic external disturbances, and it is governed by all nonzero Laplacian eigenvalues~\cite{Qi18,Bamieh12}. As we all know, $\lambda_2(G)\geq \lambda_2(H)$ when $H$ is a spanning subgraph of $G$~\cite{Godsil09,bapat}. It means that adding edges to a graph may increase its second smallest Laplacian eigenvalue. Zelazo, Schuler and Allgower~\cite{Zelazo13} provided an analytic characterization of how the addition of edges improves the convergence speed. It is well-known that $\Delta+1\leq \lambda_n \leq 2\Delta$ where $\Delta$ is the maximum vertex degree~\cite{Godsil09}. Wu and Guan~\cite{Wu07} found that we can improve the robustness to time-delay by deleting some edges linking the vertices with the maximum degree. As for network coherence, Summers \emph{et al.}~\cite{Summers15} considered how to optimize the coherence by adding some selected edges. Recently, network coherence on deterministic networks becomes a new focus. Previously studied networks include ring~\cite{Young10}, path~\cite{Young10}, star~\cite{Young10}, complete graph~\cite{Young10}, torus graph~\cite{Bamieh12}, fractal tree-like graph~\cite{Patterson14,ZW20}, Farey graph~\cite{Yi15}, web graph~\cite{Ding15}, recursive trees~\cite{Sun14,Sun16}, Koch network~\cite{Yi17}, hierarchical graph~\cite{Qi18}, Sierpinski graph~\cite{Qi18}, weighted Koch network~\cite{Dai18}, 5-rose graph~\cite{WXD19}, 4-clique motif network~\cite{Yi18}, and pseudofractal scale-free web~\cite{Yi18}. Among all these graphs, the complete graph has the optimal structure that has the best performance for noisy consensus dynamics. However, the complete graph is a dense graph which means that the communication cost of the complete graph is very high. It has been shown that networks in real-world are often sparse, small-world and scale-free. Then, Yi, Zhang and Patterson~\cite{Yi18} asked two open questions: What is the minimum scaling of the first-order coherence for sparse networks? Is this minimal scaling achieved in real scale-free networks? We will give a positive answer to these two questions in this paper.

Another interesting question about network coherence is that how network structural characteristics affect network coherence~\cite{Qi18,Yi15,Yi18}. It has been shown that the scale-free behavior and the small-world topology can significantly improve the network coherence~\cite{Yi15,Yi18}. Clearly, the star~\cite{Young10} can not be small-world for its small clustering coefficient. But the first-order coherence of a large star will converge to a small constant. The pseudofractal scale-free web~\cite{sfw} and the Farey graph~\cite{farey} are two famous small-world networks with high clustering coefficient. But we can see that the scale of the first-order coherence on the Farey graph is much larger than that on the pseudofractal scale-free web~\cite{Qi18,Yi15}. So there are some other network structural characteristics which can affect the first-order coherence. Yi, Zhang and Patterson~\cite{Qi18} argued that it would be the scale-free behavior which is absent on the Farey graph. However, the Koch network~\cite{koch} is small-world and scale-free, and the first-order coherence on the Koch network scales with the order of the network. In addition, it is clear that the complete graph and the star are not scale-free, but the first-order coherence on these two graphs are very small. So there should be something else which can affect the network coherence. In this paper, by analyzing and comparing several studied networks, we will give our answer to the above question.

The outline of this work is as follows. In Section~\ref{pre}, we present some notations and definitions in graph theory and consensus problems.
In Section~\ref{network}, we construct the hierarchical small-world network. In Section~\ref{cal}, we compute the Laplacian eigenvalues and the network coherence, and give our answers to some open questions. In Section~\ref{con}, we make a conclusion.

\section{Preliminaries}\label{pre}
Let $G=(V,E)$ be a connected and undirected graph (network). $|S|$ denotes the cardinality of the set $S$. The order (number of vertices) and the size (number of edges) of $G$ are $n=|V|$ and $m=|E|$, respectively. If $e=\{u,v\}$ is an edge of $G$, we say vertices $u,v$ are \emph{adjacent} by $e$, and $u$ is a \emph{neighbor} of $v$. Let $N_G(v)$ denote the set of neighbors of $v$ in graph $G$. The \emph{degree} of vertex $v$ in graph $G$ is given by $d_g(v)=|N_G(v)|$. We denote the maximum and minimum vertex degrees of $G$ by $\Delta$ and $\delta$, respectively. Let $S$ be a subset of vertex set $V$. $G-S$ is the graph obtained from $G$ by deleting all vertices in $S$. If $G-S$ is disconnected, we call $S$ a \emph{vertex cut set} of $G$. The \emph{vertex connectivity} $c_v$ of graph $G$ is defined as the minimum order of all vertex cut sets. Similarly, we can define the \emph{edge connectivity} $c_e$.

The \emph{density} of graph $G$  is given by~\cite{d}
\begin{equation}\label{sparse}
d=\frac{2m}{n(n-1)},
\end{equation}
Clearly, $0\leq d\leq1$. $G$ is a sparse graph if and only if $d\ll1$.
\subsection{Four graph matrices}
The \emph{adjacency matrix} of graph $G$ is a symmetric matrix $A=A(G)=[a_{i,j}]$, whose $(i,j)$-entry is
\begin{equation}\nonumber
a_{i,j}=\left\{
  \begin{array}{ll}
    1, & \hbox{if $v_i$ is adjacent with  $v_j$;} \\
    0, & \hbox{otherwise.}
  \end{array}
\right.
\end{equation}
The \emph{degree matrix} of graph $G$ is a diagonal matrix $D=D(G)=[d_{i,j}]$ where
\begin{equation*}
d_{i,j}=\left\{
  \begin{array}{ll}
    d_G(v_i), & \hbox{if $i=j$;} \\
    0, & \hbox{otherwise.}
  \end{array}
\right.
\end{equation*}

The \emph{Laplacian matrix} of graph $G$ is defined by $L=D-A$. We write $(LX)_i$ as the element corresponding to the vertex $v_i$ in the vector $LX$.

\begin{teo}\cite{Godsil09,bapat}
Let $G$ be a connected and undirected graph with vertex set $V(G)=\{v_1,v_2,\cdots,v_n\}$. $X=(x_1,x_2,\cdots,x_n)^{\top}$ is a column vector. Then the following assertions hold.

(i) $LX=\lambda X $ if and only if, for each $i$,
\begin{eqnarray}
(LX)_i&=&d_G(v_i)x_i-\sum_{v_j\in N_G(v_i)}x_j  \label{lte} \\
      &=&\lambda x_i.  \nonumber
\end{eqnarray}

(ii) The rank of $L$ is $n-1$.

(iii) The row (column) sum of $L$ is zero.

\end{teo}

According to $(iii)$, we know that  $0$ is an eigenvalue of $L$ with corresponding eigenvector $\mathbf{1}=(1,1,\cdots,1)^{\top}$. Since the rank of $L$ is $n-1$, we can write the eigenvalues of $L$ as $0=\lambda_1<\lambda_2\leq \lambda_3\leq \cdots\leq \lambda_n$. The second smallest Laplacian eigenvalue $\lambda_2$ is called the \emph{algebraic connectivity} of the graph. This concept was introduced by Fiedler~\cite{Fiedler}. A classical result on the bounds for $\lambda_2$ is given by Fiedler~\cite{Fiedler} as follows:
\begin{equation}\label{Fied}
\lambda_2\leq c_v\leq c_e\leq \delta.
\end{equation}

The \emph{transition matrix} of graph $G$ is a $n$-order matrix $P=[p_{i,j}]$ in which $p_{ij}=\frac{a_{i,j}}{d_G(v_i)}$. So $P=D^{-1}A$. Since $P$ is conjugate to the symmetric matrix $D^{-\frac{1}{2}}AD^{-\frac{1}{2}}$, all eigenvalues of $P$ are real. We denote these eigenvalues as $1=\theta_1>\theta_2\geq\cdots\geq\theta_n$.

\subsection{Consensus problems}
In this subsection, we give a simple introduction from a graph theory perspective to \emph{consensus problems}~\cite{Qi18,Patterson14}. We refer the readers to Refs.~\cite{Saber04,Qi18,Bamieh12,Patterson14} for more details. The information exchange network of a multi-agent system can be modeled via graph $G$. Each vertex of graph $G$ represents an agent, and each edge of graph $G$ represents a communication channel. Two endpoints of an edge can exchange information with each other through the communication channel. Usually, the state of the system at time $t$ is given by a column vector $X(t)=(x_1(t),x_2(t),\cdots,x_n(t))^{\top}\in \mathbb{R}^{n}$, where $x_i(t)$ denotes the state (e.g., position, velocity, temperature, etc.) of the agent $v_i$ at time $t$. Each agent can update its state according to its current state and the information received from its neighbors. Generally, the dynamic of each agent $v_i$ can be described by $\dot{x}_i(t)=u_i(t)$ where $u_i(t)$ is the \emph{consensus protocol (or algorithm)}.

\subsubsection{Consensus without communication time-delay and noise}
Olfati-Saber and Murray~\cite{Saber04} proved that if
\begin{equation}\label{pro1}
u_i(t)=\sum_{v_j\in N_G(v_i)}(x_j(t)-x_i(t)),
\end{equation}
then, the state vector $X(t)$ evolves according to the following differential equation:
\begin{equation}\label{pro11}
\dot{X}(t)=-LX(t),
\end{equation}
 and asymptotically converges to the average of the initial states (\emph{i.e.}, for each $i$, $lim_{t\rightarrow \infty}x_i(t)=\frac{1}{n}\sum_{k=1}^nx_k(0)$, where $x_k(0)$ is the initial state of agent $v_k$). This means that the system with protocol~(\ref{pro1}) can reach an average-consensus. In addition, the convergence speed of $X(t)$ can be measured by the second smallest Laplacian eigenvalue $\lambda_2$: the larger the value of $\lambda_2$, the faster the convergence speed~\cite{Qi18}.


\subsubsection{Consensus with communication time-delay}
There are some finite time lags for agents to communicate with each other in many real-world networks. Olfati-Saber and Murray~\cite{Saber04} showed that if the time delay for all pairs of agents is independent on $t$ and fixed to a small constant $\epsilon$,
\begin{equation}\label{pro2}
u_i(t)=\sum_{v_j\in N_G(v_i)}(x_j(t-\epsilon)-x_i(t-\epsilon)),
\end{equation}
then, the state vector $X(t)$ evolves according to the following delay differential equation:
\begin{equation}\label{c4}
\dot{X}(t)=-LX(t-\epsilon).
\end{equation}
In addition, $X(t)$ asymptotically converge to the average of the initial states if and only if $\epsilon$ satisfies the following condition:
\begin{equation}\label{pro3}
0< \epsilon < \epsilon_{max}=\frac{\pi}{2\lambda_n}.
\end{equation}
Eq.~(\ref{pro3}) shows that the largest Laplacian eigenvalue $\lambda_n$ is a good measure for delay robustness: the smaller the value of $\lambda_n$, the bigger the maximum delay $\epsilon_{max}$~\cite{Qi18}. Moreover, similarly to the system with protocol~(\ref{pro1}), the convergence speed of the system with protocol~(\ref{pro2}) is also determined by the second smallest Laplacian eigenvalue $\lambda_2$~\cite{Saber04,Qi18}.

\subsubsection{Consensus with white noise}

In order to capture the robustness of consensus algorithms when the agents are subject to external perturbations, Patterson and Bamieh~\cite{Patterson14} introduced a new quantity called \emph{network coherence}.

\emph{First-order network coherence:} In the first-order consensus problem, each agent has a single state $x_i(t)$. The dynamics of this system are given by~\cite{Qi18,Bamieh12,Patterson14}
\begin{equation}\label{c6}
\dot{X}(t)=-LX(t)+w(t)
\end{equation}
where $w(t) \in \mathbb{R}^n$ is the white noise.

It is interesting that if the noise $w(t)$ satisfies some particular conditions, the state of each agent $x_i(t)$ does not necessarily converge to the average-consensus, but fluctuates around the average of the current states~\cite{Qi18,Patterson14}. The variance of these fluctuations can be captured by network coherence.

\begin{defi}(Definition 2.1 of \cite{Patterson14})\label{def1}
For a connected graph $G$, the first-order network coherence $H_1$ is defined as the mean (over all vertices), steady-state variance of the deviation from the average of the current agents states,
\[
H_{1}=H_1(G)=\lim_{t\to\infty}\frac{1}{n}\sum^n_{i=1}{\bf{{var}}}
\left\{x_i(t)-\frac{1}{n}\sum^n_{j=1}x_j(t)\right\},
\]
where $\bf{var\{\cdot\}}$ denotes the variance.
\end{defi}

It is amazing for algebraic graph theorists that $H_1$ is completely determined by the $n-1$ nonzero Laplacian eigenvalues~\cite{Bamieh12}. Specifically, the first-order network coherence equals
\begin{equation}\label{h1}
H_{1}=\frac{1}{2n}\sum^{n}_{i=2}\frac{1}{\lambda_i}.
\end{equation}
Lower $H_{1}$ implies better robustness of the system irrespective of the presence of noise, i.e., vertices remain closer to consensus at the average of their current states~\cite{Qi18,Patterson14}.

\emph{Second-order noisy consensus:} In the second-order consensus problem, each agent $v_i$ has two state variables $x_i(t)$ and $y_i(t)$. Agent $v_i$ updates its state based on the states of itself and its neighbors. The state of the entire system at time $t$ is $X(t),Y(t)$, and random external disturbances enter through the $Y(t)$ terms. The dynamics of this system  are~\cite{Qi18,Patterson14}

\begin{equation}\label{c8}
\left[ \begin{array}{c}
\dot{X}(t)\\
\dot{Y}(t)
\end{array} \right]=\left[ \begin{array}{cc}
0     & I\\
-L & -L
\end{array} \right]\left[ \begin{array}{c}
X(t)\\
Y(t)
\end{array} \right]+\left[ \begin{array}{c}
0\\
I
\end{array} \right]w(t),
\end{equation}
where $w(t) \in \mathbb{R}^{n}$ is a disturbance vector with zero-mean, unit variance.

The network coherence of the second-order system~(\ref{c8}) is defined in terms of $X(t)$ only.
\begin{defi}(Definition 2.2 of \cite{Patterson14})
For a connected graph $G$, the second-order network coherence $H_{2}$ is the mean (over all vertices), steady-state variance of the deviation from the average of $X(t)$
\[
H_{2}=H_2(G)=\lim_{t\to\infty}\frac{1}{n}\sum^n_{i=1}{\bf{{var}}}
\left\{x_i(t)-\frac{1}{n}\sum^n_{j=1}x_j(t)\right\}.
\]
\end{defi}

The value of $H_2$ can also be completely determined by the nonzero Laplacian eigenvalues~\cite{Bamieh12}, specifically
\begin{equation}\label{h2}
H_{2}=\frac{1}{2n}\sum^{n}_{i=2}\frac{1}{\lambda_i^2}.
\end{equation}
A small $H_{2}$ implies that better robust to external disturbances.

\section{Network construction and properties}\label{network}
The hierarchical small-world network was introduced by Chen \emph{et al.}~\cite{ChenPhysicaA} and can be created following a recursive-modular method, see Fig.~\ref{net}. Let $M^r_g$ $(g\geq 0, r\geq 2)$ denote the network after $g$ generations of evolution. For $g=0$, the network $M_0$ is a single vertex. For $g\geq 1$, $M^r_g$ can be obtained from $r$ copies of $M^r_{g-1}$ and a new vertex by linking the new vertex to every vertex in each copy of $M^r_{g-1}$.
\begin{figure}[ht]
\begin{center}
\includegraphics[width=0.4\linewidth]{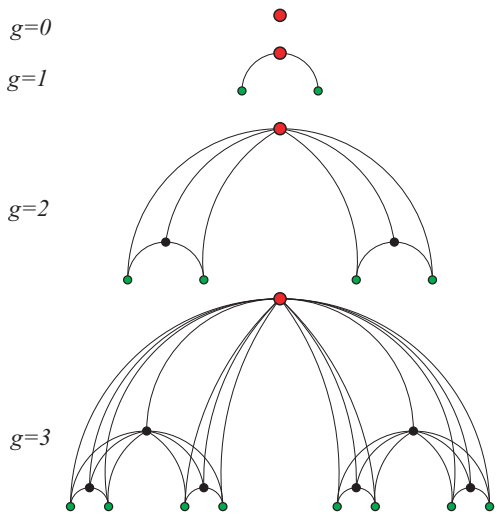}
\end{center}
\caption{Growing processes of $M_g^r$ with $r=2$. $M_g^r$ consists of $r$ copies of $M_{g-1}^r$ and a new vertex.}
\label{net}
\end{figure}

 A \emph{rooted tree} $T$ is a tree with a particular vertex $v_0$, see Fig~\ref{roottree}. We call $v_0$ the root of $T$. Let $v$ be a vertex of $T$. If $v$ has only one neighbor, $v$ is a \emph{leaf} of $T$; if $v$ has at least two neighbors, $v$ is a \emph{non-leaf vetex} of $T$. The \emph{level} of vertex $v$ in $T$ is the length of the unique path from $v_0$ to $v$. Note that the level of the root $v_0$ is $0$. The \emph{height} of rooted tree $T$ is the largest level number of all vertices. We always use a directed tree to describe a rooted tree by replacing each edge with an arc (directed edge) directing from a vertex of level $i$ to a vertex of level $i+1$. Fig.~\ref{roottree} shows a root tree of height $3$.
\begin{figure}[ht]
\begin{center}
\includegraphics[width=0.4\linewidth]{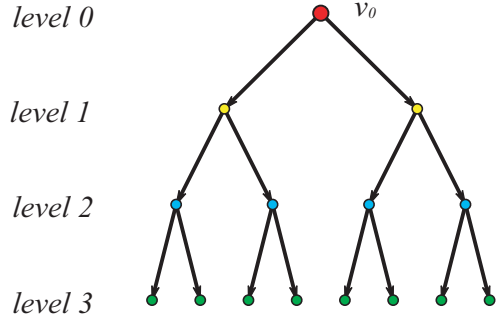}
\end{center}
\caption{A rooted tree $T$ of height $3$.}
\label{roottree}
\end{figure}
If $(u,v)$ is an arc of the rooted tree $T$, then $u$ is the \emph{parent} of $v$, and $v$ is a \emph{child} of $u$. If there is a unique directed path from a vertex $v$ to a vertex $w$, we say that $v$ is an \emph{ancestor} of $w$, and $w$ is a \emph{descendant} of $v$. For $r\geq 2$, a rooted tree is called a \emph{full $r$-ary tree} if  all leaves are in the same level and every non-leaf vertex has exactly $r$ children. The rooted tree illustrated in Fig.~\ref{roottree} is a full $2$-ary (binary) tree.
The $r$-ary tree is one of the most important data structures in computer science~\cite{Brent,bar91,Srinivas}, and it has various applications in biology~\cite{San98,Otu}, and graph theory~\cite{Fu00}. It is not difficult to see that the hierarchical network $M^r_g$ can also be obtained from a rooted tree $T_g$ of height $g$ by linking every non-leaf vertex to all its descendants, and we call the rooted tree $T_g$ the basic tree of $M^r_g$.

Let $N_g$ and $E_g$ denote the order and size of the hierarchical small-world network $M^r_g$, respectively. According to the two construction algorithms, we have
\begin{equation}\label{a1}
N_g=\frac{1}{r-1}(r^{g+1}-1)
\end{equation}
and
\begin{eqnarray*}
E_g&=&\frac{1}{(r-1)^2}\left(gr^{g+2}-gr^{g+1}-r^{g+1}+r\right)\\
   &=&\frac{g+1}{r-1}+(g-\frac{1}{r-1})N_g.
\end{eqnarray*}

According to Eq.~(\ref{sparse}), the density of $M_g^r$ is given by
\begin{eqnarray*}
d&=&\frac{2E_g}{N_g(N_g-1)}\\
&=&\frac{2(g+1)}{(r-1)N_g(N_g-1)}+\frac{2(g-\frac{1}{r-1})}{N_g-1}.
\end{eqnarray*}
If $N_g\gg1$, then $\frac{g+1}{(r-1)N_g(N_g-1)}\rightarrow 0$ and also $\frac{g-\frac{1}{r-1}}{N_g-1}\ll 1$, that is $d\ll 1$. Hence, the hierarchical small-world network $M_g^r$ is a sparse network.

In an arbitrary level $i$ of the basic tree $T_g$, there are $r^i$ vertices. We randomly choose a vertex $v$. The probability that it comes from level $i$ is
\begin{equation}\label{a4}
P(i)=\frac{r^{i}(r-1)}{r^{g+1}-1}.
\end{equation}
Since all vertices in level $i$ have the same degree $k_i=\frac{r^{g+1-i}-1}{r-1}+i-1$, and vertices in different levels have different degrees, the degree distribution $P(k)$ of the hierarchical small-world network is
\begin{equation*}
P(k)=\frac{r-1}{(r-1)(1+k-i)+1-r^{-i}}
\end{equation*}

 The degree distribution of a real-world network always follows a power-law distribution $P(k)\sim k^{-\gamma}$ with $\gamma> 1$~\cite{Bara,Albert,Dorogovtsevavp10}. For the hierarchical small-world network, according to the result in~\cite{ChenPhysicaA}, we know that $\gamma$ would approach $1$ when $M_g^r$ is large enough. That is abnormal. However, this network model exists in real life since it is a good model for the benefit transmission web of the pyramid selling~\cite{sell2} in China and many other countries.

\section{Calculations of network coherence}\label{cal}
\subsection{Eigenvalues and their corresponding eigenvectors}
Let $T_g$ be the basic tree of $M^r_g$ and $V(T_g)=V(M^r_g)=\{v_0,v_1,\cdots,v_{N_g-1}\}$. The root of $T_g$ is $v_0$. For each vertex $v_i\in T_g$, we denote the set of descendants (ancestors) of $v_i$ by $des(v_i)$ $(anc(v_i))$. Let $D_i=|des(v_i)|$ and $A_i=|anc(v_i)|$. Let $d_g(v_i)$ be the degree of $v_i$ in $M_g^r$. It is important to note that $d_g(v_i)=D_i+A_i$ for each $i$. Let $l_g(v_i)$ denote the level of $v_i$ in $T_g$. It is not difficult to see that $l_g(v_i)=A_i$.

In order to help the readers to get a direct impression of the following theorem and a better understanding of the proof, we introduce an example.
\begin{figure}[ht]
\begin{center}
\includegraphics[width=0.5\linewidth]{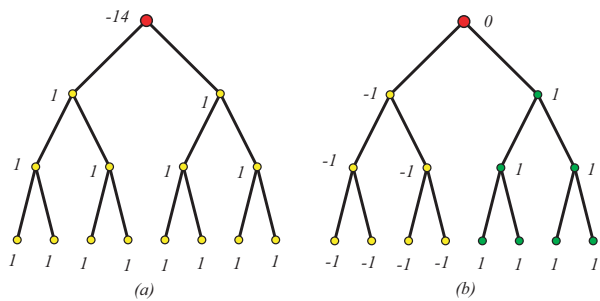}
\end{center}
\caption{(a) The eigenvector corresponding to eigenvalue $\alpha=15$; (b) The eigenvector corresponding to eigenvalue $\beta=1$.}
\label{vector2}
\end{figure}
\begin{es}
As shown in Fig.~\ref{roottree}, root $v_0$ is a non-leaf vertex of the basic tree $T_3$, and $d_3(v_0)=14,l_3(v_0)=0$. $v_1$ is the left child of root $v_0$. Let $\alpha=d_3(v_0)+1=15$, $\beta=l_3(v_0)+1=1$. Then $\alpha$ is a Laplacian eigenvalue of $M_3^2$ with corresponding eigenvector $X$ shown in Fig.~\ref{vector2}(a), because Eq.~(\ref{lte}) holds for every vertex of $M_3^2$. For instance, equations $(LX)_0=14\cdot (-14)-1\cdot 14=15\cdot(-14)=\alpha x_0$ and $(LX)_1=7\cdot 1+14-6\cdot1=15\cdot1=\alpha\cdot x_1$ show that Eq.~(\ref{lte}) holds for root $v_0$ and vertex $v_1$, respectively. Similarly, $\beta$ is also a Laplacian eigenvalue of $M_3^2$ with corresponding eigenvector $X'$ shown in Fig.~\ref{vector2}(b). For instance, equations $(LX')_0=0+7\cdot 1-7\cdot 1=0=\beta x'_0$ and $(LX')_1=7\cdot (-1)+6\cdot1=-1=\beta x'_1$ verify that Eq.~(\ref{lte}) holds for root $v_0$ and vertex $v_1$, respectively.

\end{es}

\begin{teo}\label{mainteo}
The nonzero Laplacian eigenvalues of the hierarchical small-world network $M^r_g$ are the following:

1. $d_g(v_i)+1$, repeated exactly once for each non-leaf vertex $v_i$;

2. $l_g(v_i)+1$, repeated $r-1$ times for each non-leaf vertex $v_i$.
\end{teo}
\begin{proof}
Let $L_g$ be the Laplacian matrix of $M_g^r$. As mentioned above, $0$ is a special eigenvalue of $L_g$ with corresponding eigenvector $\mathbf{1}_{N_g}=(1,1,\cdots,1)^{\top}$. Since $M^r_g$ is connected, $L_g$ has $N_g-1$ non-zero eigenvalues. It is clear that, in the full $r$-ary tree $T_g$, each vertex of level $g$ is a leaf, and each vertex of level $i$ $(i\leq g-1)$ is a non-leaf vertex. Thus $T_g$ has $\frac{r^g-1}{r-1}$ non-leaf vertices. So the total number of eigenvalues mentioned in the statement of this theorem add up to $\frac{r^g-1}{r-1}\cdot r=\frac{r^{g+1}-r}{r-1}$ which equals $N_g-1$.

{\bf Case 1:} When $\lambda_g=d_g(v_i)+1$ where $v_i$ is a non-leaf vertex in $T_g$.

Let $X_g=\left(x_0,x_1,\cdots,x_{N_g-1}\right)^{\top}$ be a column vector, and
\begin{equation*}
x_k=\left\{
  \begin{array}{ll}
    -D_i, & \hbox{if $k=i$;} \\
       1, & \hbox{if $v_k$ is a descendant of $v_i$;} \\
       0, & \hbox{otherwise.}
  \end{array}
\right.
\end{equation*}
Since the number of descendants of $v_i$ is just $D_i$, we have $\sum_{j=1}^{N_g}x_j=1\cdot D_i-D_i-0=0$. Then, $X_g$ is orthogonal to the vector $\mathbf{1}_{N_g}$. We now have to prove that $X_g$ is indeed an eigenvector corresponding to the given eigenvalue $\lambda_g=d_g(v_i)+1$. In the proof, our main tool is Eq.~(\ref{lte}).

For the vertex $v_i$, $x_i=-D_i$. $x_j=1$ if $v_j$ is a descendant of $v_i$; $x_j=0$ if $v_j$ is an ancestor of $v_i$. Thus, we have
\begin{eqnarray*}
(L_gX_g)_i&=&d_g(v_i)(-D_i)-D_i-0\\
          &=&(d_g(v_i)+1)\cdot(-D_i)\\
           &=&(\lambda_gX_g)_i.
\end{eqnarray*}

For the vertex $v_k$ which is an ancestor of $v_i$, $x_k=0$. $x_j=-D_i$ if $v_j$ is $v_i$; $x_j=1$ if $v_j$ is a descendant of $v_i$; $x_j=0$ if $v_j$ is one of other neighbors of $v_k$. Hence, we have
\begin{eqnarray*}
(L_gX_g)_k&=&0+D_i-1\cdot D_i-0\\
          &=&0\cdot(d_g(v_i)+1)\\
          &=&(\lambda_gX_g)_k.
\end{eqnarray*}

For the vertex $v_k$ which is a descendant of $v_i$, $x_k=1$. $x_j=-D_i$ if $v_j$ is $v_i$; $x_j=1$ if $v_j$ is a descendant of $v_k$; $x_j=1$ if $v_j$ is a ancestor of $v_k$ and also a descendant of $v_i$; $x_j=0$ if $v_j$ is one of other neighbors of $v_k$. It is important to note that the number of ancestors of $v_k$, which are also descendants of $v_i$, is $A_k-A_i-1$, i.e., $|anc(v_k)\cap des(v_i)|=A_k-A_i-1$. We minus $1$ here because vertex $v_i$ is not a descendant of itself, it should be removed. Also, $d_g(v_k)=D_k+A_k$ and $d_g(v_i)=D_i+A_i$. Then, we have
\begin{eqnarray*}
(L_gX_g)_k&=&d_g(v_k)+D_i-1\cdot D_k-1\cdot(A_k-A_i-1)+0\\
           &=&D_i+A_i+1\\
          &=&1\cdot(d_g(v_i)+1)\\
          &=&(\lambda_gX_g)_k.
\end{eqnarray*}

It is clear that the equation $(L_gX_g)_k=(\lambda_gX_g)_k$ holds for all other vertices. Therefore, we have proved $L_gX_g=\lambda_gX_g$.

{\bf Case 2:} When $\lambda_g=l_g(v_i)+1$ where $v_i$ is a non-leaf vertex in $T_g$.

We denote the $r$ children of $v_i$ by $c_{1},c_{2},\cdots,c_{r}$. For each $t \in \{1,2,\cdots,r\}$, let $F_{t}=c_t\cup des(c_t)$. Let $|F_{1}|=f$. Since $T_g$ is a full $r$-ary tree, we have $|F_{t}|=f$ for every $t\in \{1,2,\cdots,r\}$. For each $s\in \{2,3,\cdots,r\}$, let $X_g^s=\left(x_0^s,,x_1^s,\cdots,x_{N_g-1}^s\right)^{\top}$ be a column vector, and
\begin{equation*}
x_k^s=\left\{
  \begin{array}{ll}
    -1, & \hbox{if $v_k \in F_{1}$;} \\
     1, & \hbox{if $v_k \in F_{s}$;} \\
     0, & \hbox{otherwise.}
  \end{array}
\right.
\end{equation*}
  Since $|F_s|=|F_1|=f$, $\sum_{j=1}^{N_g}x_j^s=f\cdot1+f\cdot(-1)-0=0$. Hence, $X_g^s$ is orthogonal to the vector $\mathbf{1}_{N_g}$. We now have to prove that $X_g^s$ is indeed an eigenvector corresponding to the given eigenvalue $\lambda_g=l_g(v_i)+1$.

For the vertex $v_k$ which is an ancestor of $c_{1}$, $x_k^s=0$. $x_j^s=-1$ if $v_j\in F_1$; $x_j^s=1$ if $v_j\in F_s$; $x_j^s=0$ if $v_j$ is one of other neighbors of $v_k$.  According to Eq.~(\ref{lte}), we have
\begin{eqnarray*}
(L_gX_g^s)_k&=&0+f-f-0\\
          &=&0\cdot(l_g(v_i)+1)\\
           &=&(\lambda_gX_g^s)_k.
\end{eqnarray*}

For the vertex $v_k \in F_{1}$, $x_k^s=-1$. $x_j^s=-1$ if is $v_j$ a descendant of $v_k$; $x_j^s=-1$ if $v_j$ is an ancestor of $v_k$ and also a descendant of $v_i$; $x_j^s=0$ if $v_j$ is one of other neighbors of $v_k$. Note that, $d_g(v_k)=D_k+A_k$ and $l_g(v_i)=A_i$. According to Eq.~(\ref{lte}), we have
\begin{eqnarray*}
(L_gX_g^s)_k&=&d_g(v_k)(-1)+D_k+(A_k-A_i-1)+0\\
          &=&-(A_i+1)\\
          &=&(l_g(v_i)+1)\cdot(-1)\\
          &=&(\lambda_gX_g^s)_k.
\end{eqnarray*}

For the vertex $v_k \in F_{s}$, $x_k^s=1$. $x_j^s=1$ if $v_j$ is a descendant of $v_k$; $x_j^s=1$ if $v_j$ is an ancestor of $v_k$ and also a descendant of $v_i$; $x_j^s=0$ if $v_j$ is one of other neighbors of $v_k$. According to Eq.~(\ref{lte}), we have
\begin{eqnarray*}
(L_gX_g^s)_k&=&d_g(v_k)\cdot1-D_k-(A_k-A_i-1)-0\\
          &=&(A_i+1)\\
          &=&(l_g(v_i)+1)\cdot 1\\
          &=&(\lambda_gX_g^s)_k.
\end{eqnarray*}

It is clear that the equation $(L_gX_g^s)_k=(\lambda_gX_g^s)_k$ holds for all other vertices. Thus, we have proved that $L_gX_g^s=\lambda_gX_g^s$. So $X_g^s$ is an eigenvector corresponding to the given eigenvalue $\lambda_g=l_g(v_i)+1$. Then eigenvalue $l_g(v_i)+1$ has $r-1$ linear independent vectors $X_g^2$, $X_g^3$,$\cdots$, $X_g^r$. Therefore, the multiplicity of the eigenvalue $l_g(v_i)+1$ is $r-1$.
\end{proof}

For each $i\in \{0,1,\cdots,g-1\}$, there are $r^i$ vertices in level $i$. If $v$ is a vertex in level $i$, then $d_g(v)=\frac{r^{g-i+1}-1}{r-1}+i-1$. Hence, we have the following corollary. Here we write multiplicities as subscript for convenience and there is no confusion.

\begin{cor}\label{maincor}
The nonzero Laplacian eigenvalues of the hierarchical small-world network $M^r_g$ are $(i+1)_{(r-1)r^i}$, $(\frac{r^{g-i+1}-1}{r-1}+i)_{r^i}$, where $i\in \{0,1,\cdots,g-1\}$.
\end{cor}
\begin{es}
Table~\ref{exam} lists all the $14$ nonzero Laplacian eigenvalues of the network $M_3^2$ showed in Fig.~\ref{net}. The same result can be obtained by computing the Laplacian eigenvalues of $M_3^2$ directly.

\begin{table}[!t]
\centering
\caption{All the nonzero Laplacian eigenvalues of $M_3^2$}
\label{exam}
\begin{tabular}{c|c|c|c|c|c|c}
\toprule[1pt]
Eigenvalue                                & 1                       & 2&3&15&8&5  \\
\midrule
Multiplicity                          & 1                         &2&4&1&2&4 \\
\bottomrule[1pt]
\end{tabular}
\end{table}
\end{es}

\subsection{Convergence speed and delay robustness}
As shown in Corollary~\ref{maincor}, the second smallest Laplacian eigenvalue of the hierarchical small-world network $M_g^r$ is $\lambda_2=1$ for all $g\geq 0$ and $r\geq 2$. Thus, all these networks have the same convergence speed of the consensus protocol~(\ref{pro2}). The largest Laplacian eigenvalue of $M_g^r$ is $\lambda_{N_g}=\frac{r^{g+1}-1}{r-1}$. Fig.~\ref{maxl} shows that $\lambda_{N_g}$ is an increasing function with respect to $r$ and $g$. Therefore, the consensus protocol~(\ref{pro2}) on $M_g^r$ is more robust to delay with smaller $r$ and $g$.
\begin{figure}[ht]
\begin{center}
\includegraphics[width=0.8\linewidth]{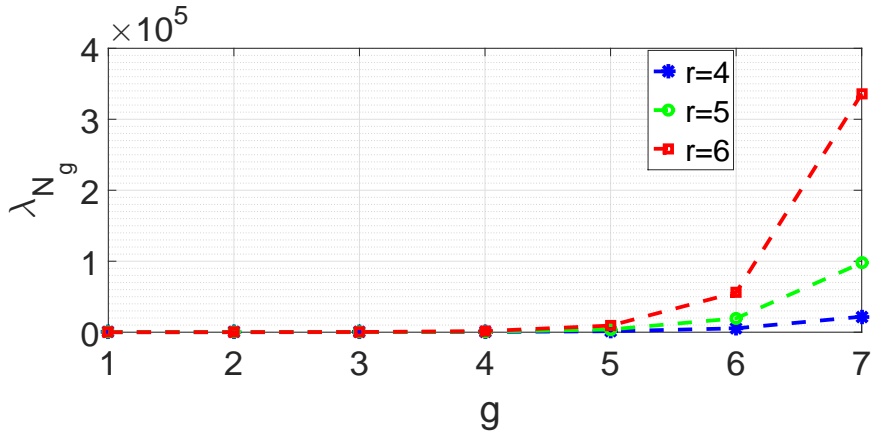}
\end{center}
\caption{The largest Laplacian eigenvalue $\lambda_{N_g}$ for hierarchical networks $M_g^r$ with various $g$ and $r$.}
\label{maxl}
\end{figure}

\begin{table*}[!t]
\centering
\caption{The Laplacian eigenvalues for some typical network structures}
\label{t1}
\begin{tabular}{c|c|c}
\toprule[1pt]
Network structure                                & $\lambda_2$                         & $\lambda_N$  \\
\midrule
Hierarchical graph $H(3,3)$~\cite{Qi18}                    & $\frac{9}{N}$                       &$\frac{2}{ln3}lnN$      \\
\midrule
Sierpinski graph $S(3,3)$~\cite{Qi18}                     & $\frac{15}{N^{log_3^5}}$            &$5$\\
\midrule
Path $P_N$~\cite{Young10}                                        &  $2-2cos(\frac{1}{N}\pi)$             &$2-2cos(\frac{N-1}{N}\pi)$      \\
\midrule
Cycle $C_N$~\cite{Young10}                                         &  $2-2cos(\frac{2}{N}\pi)$             &$2-2cos(\frac{N-1}{N}\pi)$\\
\midrule
Complete graph $K_N$~\cite{Young10}                               & $ N $                                  &$ N $                        \\
\midrule
Hierarchical SW network $M_g^2$           & 1 & $N$\\
\midrule
Star $X_N$~\cite{Young10}                                     & 1     &$N$ \\
\bottomrule[1pt]
\end{tabular}
\end{table*}

The convergence speed and delay robustness in different networks have been widely studied~\cite{Saber04,Young10,Qi18}. From the second column of Table~\ref{t1}, we find that the convergence speed of the consensus protocol~(\ref{pro2}) on the hierarchical small-world network $M_g^2$ is faster than that on other sparse graphs. At the same time, the third column shows that the consensus protocol~(\ref{pro2}) on other sparse graphs are much more robust to delay than that on the hierarchical small-world network $M_g^2$. As proved by Olfati-Saber and Murray~\cite{Saber04}, there is a tradeoff between convergence speed and delay robustness. They also claimed another tradeoff between high convergence speed and low communication cost. Here we can see this tradeoff by comparing the hierarchical small-world network $M_g^2$ with the complete graph $K_N$. The complete graph has much more edges than the hierarchical small-world network $M_g^2$ does. So the convergence speed of the consensus protocol~(\ref{pro2}) on the complete graph $K_N$ is faster than that on the hierarchical small-world network $M_g^2$. From Table~\ref{t1} we know that the complete graph $K_N$ and the hierarchical small-world network $M_g^2$ have the same delay robustness. It is clear that these two networks have the same maximum degree $N-1$. We can see that all vertices in the complete graph $K_N$ have the maximum degree, while there is only one vertex in $M_g^2$ with the maximum degree. So the delay robustness of the consensus protocol~(\ref{pro2}) is independent with the number of vertices with the maximum degree.
\subsection{Network coherence}

In~\cite{Yi18}, Yi, Zhang and Pattersom proposed two open questions: What is the minimum scaling of $H_1$ for sparse networks? Is this minimal scaling achieved in real scale-free networks? Now we want to answer these two questions.

According to Corollary~\ref{maincor} and Eqs.~(\ref{h1}) and (\ref{h2}), the following two theorems can be easily observed.
\begin{teo}\label{first}
For the hierarchical small-world network $M_g^r$, the first-order coherence is
\begin{eqnarray*}
H_1&=&\frac{(r-1)^2}{2(r^{g+1}-1)}\sum_{i=0}^{g-1}r^i(\frac{1}{r^{g-i+1}+i(r-1)-1}\\
&+&\frac{1}{i+1}).
\end{eqnarray*}
\end{teo}
\begin{teo}\label{second}
For the hierarchical small-world network $M_g^r$, the second-order coherence is
\begin{eqnarray*}
H_2&=&\frac{(r-1)^2}{2(r^{g+1}-1)}\sum_{i=0}^{g-1}r^i(\frac{r-1}{(r^{g-i+1}+i(r-1)-1)^2}\\
&+&\frac{1}{(i+1)^2}).
\end{eqnarray*}
\end{teo}
\begin{figure}[ht]
\begin{center}
\includegraphics[width=0.9\linewidth]{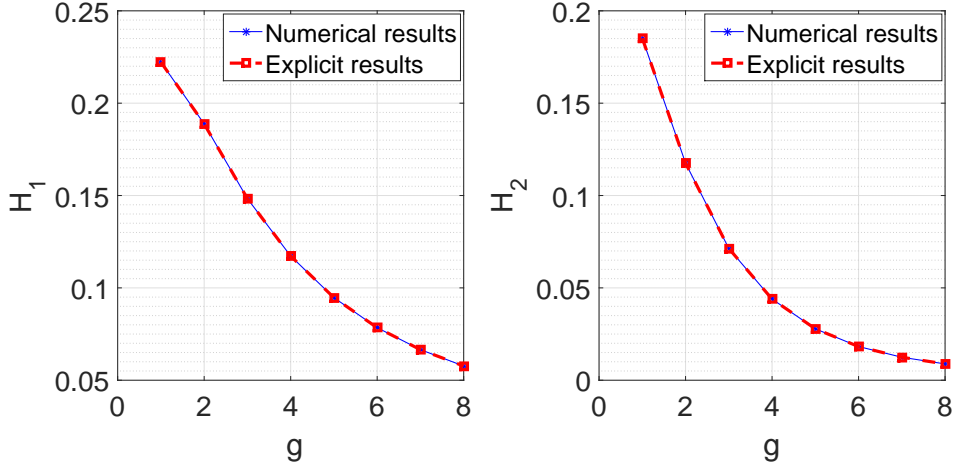}
\end{center}
\caption{The numerical results coincide with the theoretical results when $r=2$.}
\label{num}
\end{figure}
\begin{figure}[ht]
\begin{center}
\includegraphics[width=0.9\linewidth]{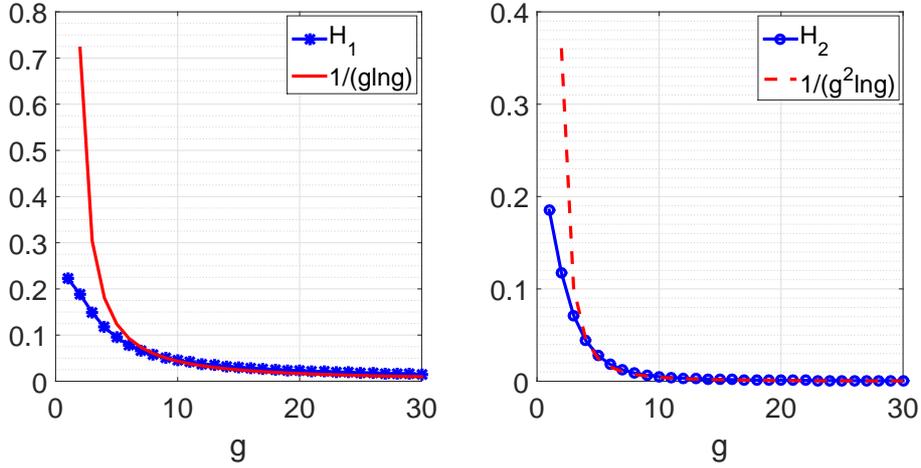}
\end{center}
\caption{Scales of $H_1$ and $H_2$.}
\label{scale}
\end{figure}

Fig.~\ref{num} shows the theoretical results coincide with the numerical results when $r=2$. The theoretical values are obtained from Theorem~\ref{first} and Theorem~\ref{second}. The numerical results are derived from Eqs.~(\ref{h1}) and~(\ref{h2}) by calculating the corresponding Laplacian eigenvalues directly. Since $N_g=\frac{r^{g+1}-1}{r-1}$, we have $g=\frac{ln((r-1)N_g+1)}{lnr}-1$. Thus, as $N_g\rightarrow \infty$, from the numerical results showed in Fig.~\ref{scale} we find that the scaling of the network coherence with network order $N$, that is, $H_1\sim \frac{1}{lnNlnlnN}$ and $H_2\sim \frac{1}{(lnN)^2lnlnN}$. This result shows that the hierarchical small-world network $M_g^r$ has the best performance for noisy consensus dynamics among sparse graphs. Therefore, we have provided an answer to the above two open questions.

\subsection{How the structural characteristics affect the network coherence}
\begin{table*}[!t]
\centering
\caption{Scalings of the first-order coherence for some typical networks.}
\label{t2}
\begin{tabular}{c|c|c|c}
\toprule[1pt]
Network                                          & $\Delta$   &  $\mu$     &  $H_1$                   \\
\midrule
Hierarchical SW network                            &    $N$         & 2       &  $(lnNln(lnN))^{-1}$       \\
\midrule
Star graph~\cite{Young10}                        &    $N$          & 2     & $1$           \\
\midrule
Pseudofractal scale-free web~\cite{Yi18}         &    $N$         &$lnN$      & $1$        \\
\midrule
4-clique motif network~\cite{Yi18}               &     $N$          &$lnN$     & $1$ \\
\midrule
Koch graph~\cite{Yi17}                           &     $N$           &$lnN$    &  $lnN$\\
\midrule
Farey graph~\cite{Yi15}                          &     $lnN$         &$lnN$    & $lnN$    \\
\bottomrule[1pt]
\end{tabular}
\end{table*}

Since the scale of the second-order coherence can be predicted by the first-order coherence, we now just study the effect of the network topologies on the first-order network coherence which has been extensively studied in previous works~\cite{Young10,Patterson14,Yi18}. Fig.~\ref{r} shows that the differences between the network coherence with different $r$ are very small when $g$ is large enough. Hence, the effect of the parameter $r$ is very limited.

\begin{figure}[ht]
\begin{center}
\includegraphics[width=0.9\linewidth]{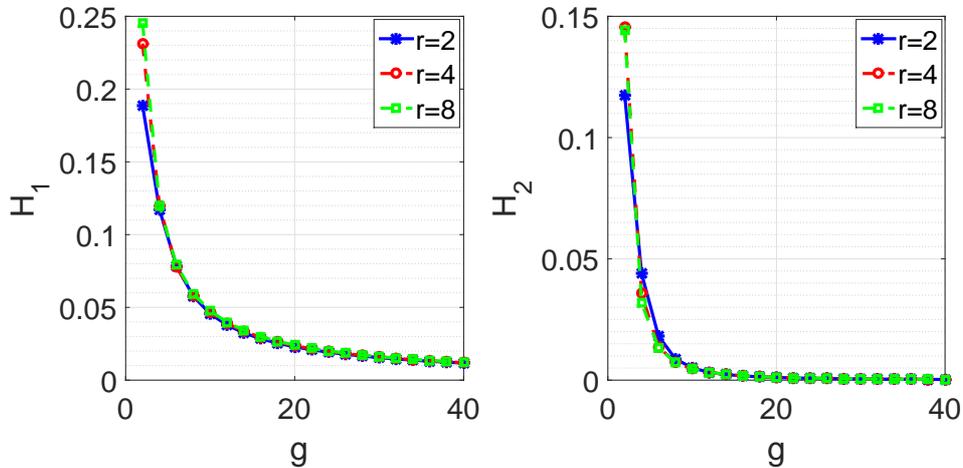}
\end{center}
\caption{Network coherence under different values of $r$.}
\label{r}
\end{figure}

Yi, Zhang and Patterson~\cite{Yi18} have given two bounds for the fist-order coherence in terms of the average path length $\mu$ and the average degree $\langle k\rangle$ of a graph. They proved that
 \begin{equation*}
 \frac{1}{2\langle k\rangle}\leq H_1\leq \frac{1}{4}\mu.
 \end{equation*}
 So networks with constant $\mu$ must have limited $H_1$, and the first-order coherence of networks with constant $\langle k\rangle$ can not be $0$. For example, for a large network order, the complete graph $K_N$, the star~$X_N$ and the hierarchical small-world network have constant $H_1$ and the first-order coherence of the star graph $X_N$ is a non-zero constant, see Table~\ref{t2}.
But when $\mu$ is an increasing function of the order $N$, we can not estimate the scale of $H_1$ from the upper bound.

From Eq.~(\ref{h1}), we can find that the algebraic connectivity $\lambda_2$ plays an important role in determining the value of $H_1$. If $\lambda_2$ is close to $0$, $H_1$ will become very large. From the classical Fiedler inequality~(\ref{Fied}), we know that $\lambda_2$ is bounded by the vertex connectivity $c_v$ and the edge connectivity $c_e$ which measure respectively the robustness to vertex and edge failure. So graphs with large vertex connectivity and edge connectivity tend to have small $H_1$. In other words, high robustness to vertex and edge failures mean high robustness against uncertain disturbance.
For example, the Koch network has the same scale of the maximum degree $\Delta$ and the average path length $\mu$ as the pseudofractal scale-free web, but the scale of the first-order coherence $H_1$ in the Koch network is very high, see Table~\ref{t2}. This is because the Koch network is not robust to vertex failure.

The Kirchhoff index $R(G)=N\sum_{i=2}^N\frac{1}{\lambda_i}$ is a famous and important graph invariant~\cite{palacios}. The relation between $R(G)$ and $H_1$ is given by $H_1=\frac{R(G)}{2N^2}$. We have the following bounds for $R(G)$~\cite{palacios},
\begin{equation*}
\frac{N}{\Delta}\sum_{i=2}^N\frac{1}{1-\theta_i}\leq R(G) \leq \frac{N}{\delta}\sum_{i=2}^N\frac{1}{1-\theta_i},
\end{equation*}
where $\theta_i$ are the eigenvalues of the transition matrix $P$ defined in Section~\ref{pre}. Thus, we have two new bounds for $H_1$,
\begin{equation*}
\frac{1}{2N\Delta}\sum_{i=2}^N\frac{1}{1-\theta_i}\leq H_1 \leq \frac{1}{2N\delta}\sum_{i=2}^N\frac{1}{1-\theta_i}.
\end{equation*}
Hence, the maximum degree is a good predictor for $H_1$. Small $\Delta$ means large $H_1$. For example, the Farey graph has the same scale of the average path length as the pseudofractal scale-free web, but the scale of $H_1$ in the Farey graph is much larger. This is due to the scale of the maximum degree in the Farey graph is low, see Table~\ref{t2}.

\section{Conclusion}\label{con}
In this paper, by applying the spectral graph theory, we studied three important aspects of consensus problems in a hierarchical small-world network which is a real-life network model. Compared with several previous studies, consensus algorithm in the hierarchical small-world network converges faster but less robust to communication time delay. It is worth mentioning that the hierarchical small-world network has optimal network coherence which captures the robustness of consensus algorithms when the agents are subject to external perturbations. These results provide a positive answer to two open questions of Yi, Zhang and Patterson~\cite{Yi18}. Finally, we argue that some particular network structure characteristics, such as large maximum degree, small average path length, and large vertex and edge connectivity, are responsible for the strong robustness with respect to external perturbations.

\section*{Acknowledgment}

This work was supported by Normandie region France and the XTerm ERDF project (European Regional Development Fund) on Complex Networks and Applications, the National Natural Science Foundation of China (No. 11971164), the Natural Science Foundation of Hunan Province (Nos. 2018JJ3255), the Scientific Research Fund of Hunan Province Education Department (Nos. 19B313,19B319).


\begin{thebibliography}{99}

\bibitem{DeGroot74}
M. H. DeGroot,
Reaching a consensus,
\emph{J. Amer. Stat. Assoc.} {\bf 69}(345) (1974) 118-121.


\bibitem{Saber04}
R. Olfati-Saber and R. M. Murray,
Consensus problems in networks of agents with switching topology and time-delays,
\emph{IEEE Trans. Autom. Control} {\bf 49}(9) (2004) 1520-1533.



\bibitem{Korniss03}
G. Korniss, M. A. Novotny, H. Guclu, Z. Toroczkai and P. A. Rikvold,
Suppressing roughness of virtual times in parallel discrete-event simulations,
\emph{Science} {\bf 299}(5607) (2003) 677-679.



\bibitem{Li06}
Q. Li and D. Rus,
Global clock synchronization in sensor networks,
\emph{IEEE Trans. Comput.} {\bf 55}(2) (2006) 214-226.


\bibitem{Article1}
Y. Liao, M. Maama, and M.A. Aziz-Alaoui, 
Optimal networks for exact controllability,
\emph{International Journal of Modern Physics C.} {\bf 31}(10) (2020) 2050144.




\bibitem{Sumpter08}
D. Sumpter, J. Krause, R. James, I. Couzin and A. Ward,
Consensus decision making by fish,
\emph{Curr. Bio.} {\bf 18}(22) (2008) 1773-1777.

\bibitem{Nabet09}
B. Nabet, N. Leonard, I. Couzin and S. Levin,
Dynamics of decision making in animal group motion,
\emph{J. Nonlinear. Sci.} {\bf 19}(4) (2009) 399-435.

\bibitem{Giraldo16}
L. F. Giraldo and K. M. Passino,
Dynamic task performance, cohesion, and communications in human groups,
\emph{IEEE Trans. Cybern.} {\bf 46}(10) (2016) 2207-2219.


\bibitem{WXD19}
X. Q. Wang, H. L. Xu and M. F. Dai,
First-order network coherence in $5$-rose graphs,
\emph{Physica A} {\bf 527} (2019) 121129.


\bibitem{WXM19}
X. Q. Wang, H. L. Xu and C. B. Ma,
Consensus problems in weighted hierarchical graphs,
\emph{Fractals} {\bf 27}(6) (2019) 1950086.

\bibitem{DZHL20}
M. F. Dai, J. Zhu, F. Huang, Y. Li, L. H. Zhu and W. Y. Su,
Coherence analysis for iterated line graphs of multi-subdivision graph,
\emph{Fractals} {\bf 28}(4) (2020) 2050067.



\bibitem{Qi18}
Y. Qi, Z. Z. Zhang, Y. H. Yi and H. Li,
Consensus in self-similar hierarchical graphs and sierpinski graphs: convergence speed, delay robustness, and coherence,
\emph{IEEE Trans. Cybern.} {\bf 49}(2) (2019) 592-603.



\bibitem{Bamieh12}
B. Bamieh, M. Jovanovi\'c, P. Mitra and S. Patterson,
Coherence in largescale networks: Dimension dependent limitations of local feedback,
\emph{IEEE Trans. Autom. Control} {\bf 57}(9) (2012) 2235-2249.


\bibitem{Godsil09}
C. Godsil and G. Royle,
\emph{Algebraic graph theory}
(Springer-Verlag, New York, 2001).

\bibitem{bapat}
R. B. Bapat,
\emph{Graph and Matrices}
(Springer, London, 2010).


\bibitem{Zelazo13}
D. Zelazo, S. Schuler and F. Allgower,
Performance and design of cycles in consensus networks,
\emph{Syst. Cont. Lett.} {\bf 62}(1) (2013) 85-96.

\bibitem{Wu07}
Z. P. Wu and Z. H. Guan,
Time-delay robustness of consensus problems in regular and complex networks,
\emph{Int. J. Mod. Phys. C} {\bf 18}(8) (2007) 1339-1350.

\bibitem{Summers15}
T. Summers, I. Shames, J. Lygeros and F. Dorfler,
Topology design for optimal network coherence,
in \emph{European Control Conf.} (IEEE, New York, 2015), pp. 575-580.



\bibitem{Young10}
G. F. Young, L. Scardovi and N. E. Leonard,
Robustness of noisy consensus dynamics with directed communication,
in \emph{Proc. 2010 American Control Conf.} (IEEE, New York, 2010) pp. 6312-6317.



\bibitem{Patterson14}
S. Patterson and B. Bamieh,
Consensus and coherence in fractal networks,
\emph{IEEE Trans. Control Netw. Syst.} {\bf 1}(4) (2014) 338-348.

\bibitem{ZW20}
Z. Z. Zhang and B. Wu,
Spectral analysis and consensus problems for a class of fractal network models,
\emph{Physica Scripta} {\bf 95}(8) (2020) 085210.


\bibitem{Article3}
Y. Liao, M. Maama, and M.A. Aziz-Alaoui, 
The number of spanning trees of a hybrid network created by inner-outer iteration,
\emph{Physica Scripta} {\bf 94}(10) (2019) 105205.



\bibitem{Yi15}
Y. H. Yi, Z. Z. Zhang, Y. Lin and G. R. Chen,
Small-world topology can significantly improve the performance of noisy consensus in a complex network,
\emph{Comput. J.} {\bf 58}(12) (2015) 3242-3254.

\bibitem{Ding15}
Q. Y. Ding, W. G. Sun and F. Y. Chen,
Network coherence in the web graphs,
\emph{Commun. Nonlinear Sci. Numer. Simulat.} {\bf 27}(1) (2015) 228-236.

\bibitem{Sun14}
W. G. Sun, Q. Y. Ding, J. Y. Zhang and F. Y. Chen,
Coherence in a family of tree networks with an application of Laplacian spectrum,
\emph{Chaos} {\bf 24}(4) (2014) 043112.


\bibitem{Sun16}
W. G. Sun, T. F. Xuan and S. Qin,
Laplacian spectrum of a family of recursive trees and its applications in network coherence,
\emph{J. Stat. Mech.} (2016) 063205.


\bibitem{Yi17}
Y. H. Yi, Z. Z. Zhang, L. R. Shan and G. R. Chen,
Robustness of first- and second-order consensus algorithms for a noisy scale-free samll-world Koch network,
\emph{IEEE Trans. Control Syst. Technol.} {\bf 25}(1) (2017) 342-350.



\bibitem{Dai18}
M. F. Dai, J. J. He, Y. Zong, T. T. Ju, Y. Sun and W. Y. Su,
Coherence analysis of a class of weighted networks,
\emph{Chaos} {\bf 28}(4) (2018) 043110.





\bibitem{Yi18}
Y. H. Yi, Z. Z. Zhang and S. Patterson,
Scale-free loopy structure is resistant to noise in consensus dynamics in complex networks,
\emph{IEEE Trans. Cybern.} {\bf 50}(1) (2020) 190-200.


\bibitem{sfw}
S. N. Dorogovtsev, A. V. Goltsev and J. F. F. Mendes,
Pseudofractal scale-free web,
\emph{Phys. Rev. E} {\bf 65}(6) (2002) 066122.



\bibitem{farey}
Z. Z. Zhang and F. Comellas,
Farey graphs as models for complex networks,
\emph{Theoret. Comput. Sci.} {\bf 412}(8) (2011) 865-875.



\bibitem{koch}
Z. Z. Zhang, S. Y. Gao, L. C. Chen, S. G. Zhou, H. J. Zhang and J. H. Guan,
Mapping Koch curves into scale-free small-world networks,
\emph{J. Phys. A} {\bf 43}(39) (2010) 395101.

\bibitem{d}
P. J. Laurienti, K. E. Joyce, Q. K. Telesford, J. H. Burdette and S. Hayasaka,
Universal fractal scaling of self-organized networks,
\emph{Physica A} {\bf 390}(20) (2011) 3608-3613.



\bibitem{Fiedler}
M. Fiedler,
Algebraic connectivity of graphs,
\emph{Czech. Math. J.} {\bf 23}(98) (1973) 298-305.



\bibitem{ChenPhysicaA}
M. Chen, B. M. Yu, P. Xu and J. Chen,
A new deterministic complex network model with hierarchical structure,
\emph{Physica A} {\bf 385}(2) (2007) 707-717.

\bibitem{Brent}
R. P. Brent and H. T. Kung,
On the area of binary tree layouts,
\emph{Inform. Process. Lett.} {\bf 11}(1) (1980) 46-48.



\bibitem{bar91}
D. R. Baronaigien,
A loopless algorithnm for generating binary tree sequences,
\emph{Inform. Process. Lett.} {\bf 39}(4) (1991) 189-194.

\bibitem{Srinivas}
S. Srinivas and N. N. Biswas,
Design and analysis of a generalized architecture for reconfigurable $m$-ary tree structures,
\emph{IEEE Trans. Comput.} {\bf 41}(11) (1992) 1465-1478.


\bibitem{San98}
D. Sankoff and M. Blanchette,
Multiple genome rearrangement and breakpoint phylogeny,
\emph{J. Comput. Biol.} {\bf 5}(3) (1998) 555-570.

\bibitem{Otu}
H. H. Otu and K. Sayood,
A new sequence distance measure for phylogenetic tree construction,
\emph{Bioinformatics} {\bf 19}(16) (2003) 2122-2130.

\bibitem{Fu00}
H. L. Fu and C. L. Shiue,
The optimal pebbling number of the complete $m$-ary tree,
\emph{Discrete Math.} {\bf 222}(1) (2000) 89-100.


\bibitem{Bara}
A.-L. Barab\'asi and R. Albert,
Emergence of scaling in random networks,
\emph{Science} {\bf 286}(5439) (1999) 509-512.


\bibitem{Article2}
M. Maama, B. Ambrosio, M.A. Aziz-Alaoui and S.M. Mintchev,
Emergent Properties in a V1-Inspired Network of Hodgkin-Huxley Neurons,
\emph{arXiv preprint } arXiv:2004.10656.

\bibitem{Albert}
R. Albert and A.-L. Barab\'asi,
Statistical mechanics of complex networks,
\emph{Rev. Mod. Phys.} {\bf 74}(1) (2002) 47-97.



\bibitem{Dorogovtsevavp10}
S. N. Dorogovtse and J. F. F. Mendes,
Evolution of networks,
\emph{Adv. Phys.} {\bf 51}(4) (2002) 1079-1187.


\bibitem{sell2}
R. Sarker,
Pyramid selling,
\emph{J. Financ. Crime} {\bf 3}(3) (1996) 266-268.




\bibitem{palacios}
J. L. Palacios and J. M. Renom,
Broder and Karlin's formula for hitting times and the Kirchhoff index,
\emph{Int. J. Quantum Chem.} {\bf 111}(1) (2011) 35-39.



\end{thebibliography}
\end{document}